\documentclass[reqno,10pt]{amsart}

\usepackage{amsmath, amsthm, amsfonts, amssymb}

\usepackage[applemac]{inputenc}

\theoremstyle{plain}
\newtheorem{theorem}{Theorem}
\newtheorem{proposition}[theorem]{Proposition}
\newtheorem{lemma}[theorem]{Lemma}
\newtheorem{corollary}[theorem]{Corollary}

\theoremstyle{definition}

\theoremstyle{remark}
\newtheorem{remark}[theorem]{Remark}

\DeclareMathOperator{\res}{Res}

\def\Z{\mathbb{Z}}	
\def\C{\mathbb{C}}	
	
\renewcommand{\leq}{\leqslant} 		
\renewcommand{\geq}{\geqslant}

\newcommand{\nn}{\nonumber}

\def\cC{\mathcal{C}}
\def\cD{\mathcal{D}}
\def\cP{\mathcal{P}}

\begin{document}

\title[Hirota for EBTH]{Hirota equations for the extended bigraded Toda hierarchy and the total descendent potential of $\C P^1$ orbifolds}

\author{Guido Carlet}
\address{Korteweg-de Vries Institute for Mathematics, University of Amsterdam, P.O. Box 94248, 1090 GE Amsterdam, The Netherlands.}
\email{guido.carlet@uva.nl}

\author{Johan van de Leur}
\address{Department of Mathematics, 
University of Utrecht, P.O. Box 80010,
3508 TA Utrecht, 
The Netherlands.}
\email{j.w.vandeleur@uu.nl}



\date{4/4/2013}

\begin{abstract}
We prove that the Hirota quadratic equations of Milanov and Tseng define an integrable hierarchy which is equivalent to the extended bigraded Toda hierarchy. 
In particular this proves a conjecture of Milanov-Tseng that relates the total descendent potential of the orbifold $C_{k,m}$ with a tau function of  the bigraded Toda hierarchy. 
\end{abstract}
\maketitle


\section*{Introduction}
The $(k,m)$-extended bigraded Toda hierarchy (EBTH), where $k$, $m$ are two positive integers, was introduced in~\cite{Car06} as a generalization of the extended Toda hierarchy~\cite{CDZ} with $k+m$ dependent variables.
The main motivation was the attempt to define an integrable hierarchy that would encode the relations between the Gromov-Witten invariants of certain $\C P^1$ orbifolds, in analogy with the fact that the Gromov-Witten potential of $\C P^1$ is actually a tau function of the extended Toda hierarchy~\cite{Ge, DZ04, Mil06}.

Indeed the general problem of associating an explicit integrable hierarchy to the Gromov-Witten theory of a given target space $X$ has been solved only in a small number of examples. These include, beyond the $\C P^1$ case just mentioned, the $X=\mathrm{pt}$ case where the relevant integrable system, according to the Kontsevich-Witten theorem, is the KdV hierarchy, and few other cases mostly related to the equivariant version of the theory under a  complex torus action. In this context, recently some progress has been made in the case of the local Gromov-Witten theory of $\C P^1$, in relation with the Ablowitz-Ladik hierarchy~\cite{BCR}.

In a further important case Milanov and Tseng in~\cite{MT08}, considered the orbifold $C_{k,m}$ obtained from $\C P^1$ by adding two orbifold points of order $k$, $m$ respectively. They proved that the orbifold quantum cohomology ring of $C_{k,m}$ coincides with Frobenius manifold $M_{k,m}$ of degree $(k,m)$ Laurent polynomials~\cite{DZ98}. Moreover they showed that the Givental total descendent potential $D^{M_{k,m}}$ associated with $M_{k,m}$ satisfies an Hirota quadratic equation and conjectured that such equation would be equivalent to the $(k,m)$-extended bigraded Toda hierarchy. 
In this work we indeed show that the Hirota quadratic equations of Milanov and Tseng are equivalent to the extended bigraded Toda hierarchy as formulated in~\cite{Car06}.
Since, up to some details mentioned in~\cite{MT08}, the potential $D^{M_{k,m}}$ coincides with the generating function of Gromov-Witten invariants of $C_{k,m}$, one concludes that the such generating function is a tau function of the $(k,m)$-extended bigraded Toda hierarchy. 

The extended bigraded Toda hierarchy can be thought of as the analogue, in the 2D Toda hierarchy~\cite{UT84} world, of the Gelfand-Dickey reductions of the KP hierarchy. 
However, only the `standard' flows of the EBTH are obtained by restriction of the 2D Toda flows, while the so-called `logarithmic' flows have to be introduced independently, and can be defined only when the discrete space variable is replaced with a continuous one. 

The fact that the logarithmic flows of the EBTH do not originate as restrictions of the 2D Toda flows, points to the existence of a larger hierarchy, which we might call `extended 2D Toda', defined for a continuous space variable, which should include the usual 2D Toda flows and contain extra flows of logarithmic type. At the dispersionless level, such extension has been recently found in~\cite{CM}, as the principal hierarchy associated with the infinite-dimensional Frobenius manifold discovered in~\cite{CDM}. The definition of a suitable dispersive version of such hierarchy is still an open problem.

The presence of the logarithmic flows makes non-trivial to generalize some well-known construction for the Gelfand-Dickey, KP or 2D Toda hierarchies to the extended bigraded Toda.
In the case of the extended Toda hierarchy, the problem of finding Hirota quadratic equations has been originally solved  by Milanov~\cite{Mil07}. 
The main feature of his construction was the use of vertex operators with values in the algebra of formal differential operators in the space variable $x$. From the point of view of the Hirota equations, such operators are needed to cancel the multivaluedness of the logarithms appearing in the vertex operators. 
Recently a more familiar version of the Hirota quadratic equations of ETH was  suggested~\cite{Tak10}, which does not require the algebra of differential operators in $x$. In the last section we consider the equivalence of such different formulations.

The paper is organized as follows: in section 1 we recall some facts from~\cite{MT08}, mainly to fix the notations and state the Milanov-Tseng form of the Hirota equations. In particular we give the definition of the Frobenius manifold $M_{k,m}$ and independently compute the periods which enter in the definition of the vertex operators.  In section 2 we first rewrite the Hirota equations in a more standard form,  then we express them in terms of difference operators. A straightforward analysis then allows us to derive the Sato and Lax equations of the extended bigraded Toda hierarchy. 
In section 3 we show that the Hirota equations are actually equivalent to the Sato equations, namely we show that given a solution of the Sato equations we can construct from it a tau function which satisfies the Hirota equations. 
Finally, in section 4, we comment on alternative formulations of the Hirota equations for EBTH.

\subsection*{Acknowledgements}
The authors would like to thank B.~Bakalov for his insight on the alternative formulation of the HQE appearing in section 4. G.C. would like to acknowledge the Department of Mathematics of Università Milano-Bicocca for the support during the period when this work has been carried out and NWO for a travel grant to Utrecht University.

\section{The Hirota quadratic equations for the total descendent potential of $M_{k,m}$}

In this section we will recall some material from~\cite{MT08}, mainly to fix notations.  We also give a slightly different derivation of the classical limit of the periods appearing in the vertex operators. 

\subsection{The Frobenius manifold $M_{k,m}$}
\label{frob}

Let $M_{k,m}$ be the Frobenius manifold on the space of trigonometric Laurent polynomials of degree $(k,m)$, i.e.
\begin{equation} \label{LGpot}
\lambda (\zeta) = \zeta^k + u_1 \zeta^{k-1}+ \dots + u_{k+m} \zeta^{-m},
\end{equation}
as defined in~\cite{DZ98}. The identification of $T_\lambda M_{k,m}$ with $\C[\zeta,\zeta^{-1}]/(\partial_\zeta \lambda )$ induces on the tangent bundle an associative commutative product. The flat metric is defined by the residue pairing
\begin{equation}
<\partial', \partial''> = \underset{\substack{d\lambda=0 \\ |\lambda|<\infty}}{\res} \frac{\partial' \lambda(\zeta) \partial'' \lambda(\zeta)}{\lambda'(\zeta)}\frac{d\zeta}{\zeta^2} .
\end{equation}
Recall that the flat coordinates are defined in terms of residues as 
\begin{equation}
t^\alpha = -\frac{k}{\alpha} \res_{\zeta=\infty} \lambda(\zeta)^{\frac{\alpha}{k}} \ \frac{d\zeta}{\zeta}, \quad
t^{k+m-\beta} = \frac{m}{\beta} \res_{\zeta=0} \lambda(\zeta)^{\frac{\beta}{m}} \ \frac{d\zeta}{\zeta}, \quad 
\end{equation}
for $\alpha=1,\dots,k-1$, $\beta=1,\dots,m$ and $t^{k+m}= \log (Q^m u_{k+m})$.
The only non-zero entries of the metric in flat coordinates are
\begin{equation}
<\frac{\partial}{\partial t^\alpha},\frac{\partial}{\partial t^{k-\alpha}}>=\frac1k, \quad 
<\frac{\partial}{\partial t^{k+\beta}},\frac{\partial}{\partial t^{k+m-\beta}}>=\frac1m
\end{equation}
for $\alpha=1,\dots,k-1$, $\beta=0,\dots,m$.
The unity and Euler vector fields in flat coordinates are given by $e=\frac{\partial}{\partial t^k}$ and 
\begin{equation}
E=  \sum_{\alpha=1}^k \frac{\alpha}{k} t^\alpha \frac{\partial }{\partial t^\alpha}  + \sum_{\beta=1}^{m-1} \left( 1- \frac{\beta}m \right) t^{k+\beta} \frac{\partial }{\partial t^{k+\beta}} + \left( \frac1k + \frac1m \right) m \frac{\partial }{\partial t^{m+k}} .
\end{equation}

\subsection{Periods}
\label{vertex}
The vertex operators appearing in the Hirota quadratic equations are defined in terms of the ``classical limit'' of certain periods of the superpotential $\lambda(\zeta)$. 

Let $\Delta \subset M_{k,m} \times \C$ be the discriminant, the set of points $(t_0,\lambda_0)$ at which the preimage $\lambda^{-1}(\lambda_0)$ is singular, i.e. is given by less than $k+m$ distinct points. Let $\zeta_a(\lambda)$ denote one of such points. 

The periods $I^{(l)}_a(t,\lambda)$, $l\in\Z$ are multivalued functions on $(M_{k,m} \times \C )\backslash \Delta$ with values in $H:=T M_{k,m}$, defined by
\begin{equation}
\label{period}
<I_a^{(-p)}(\lambda,t),\frac{\partial }{\partial t^\alpha} >  = - \frac{\partial}{\partial t^\alpha} 
\left[ d^{-1} \left( \frac{(\lambda-\lambda(\zeta))^p}{p!} \frac{d\zeta}\zeta \right) \right]_{\zeta=\zeta_a(\lambda)} , \quad p\geq0 .
\end{equation}
Here the formal integration is defined as $d^{-1}( \zeta^s d\zeta)=(s+1)^{-1} \zeta^{s+1}$ for $s\not=-1$ and $d^{-1}(\zeta^{-1}d\zeta)=\log\zeta$.
The relation $\partial_\lambda I^{(p)}_a = I^{(p+1)}_a$, which can be easily verified for negative $p$, serves as a definition for the $p>0$ periods.

The classical limit mentioned above has to be performed as follows: first one considers the asymptotic expansion of the above expression for $\lambda\sim\infty$, and observes that the coefficients in such expansion are polynomials in the flat coordinates\footnote{In this section we set the Novikov parameter $Q$ to $1$, for the sake of simplicity. One can recover the dependence on $Q$ by shifting $t^{k+m}$ by $m \log Q$.}  $t^1, \dots , t^{k+m}$ and in $e^{t^{k+m}/m}$.  Then one sets to zero $t^\alpha$ and $e^{t^{k+m}/m}$, which amounts to taking the constant coefficient of such polynomials.
In contrast with what happens in the $A_{n+1}$ case~\cite{Giv}, note that here ``limit'' does not correspond to an actual limit to a point of the Frobenius manifold, since $t^{k+m}$ and $e^{t^{k+m}}$ cannot evidently be set to zero at the same time. This phenomenon can be traced back to the resonant spectrum of the Frobenius manifold $M_{k,m}$.

The classical limit of the periods $I_a^{(p)}$ has been computed in~\cite{MT08} using the fact that they satisfy certain differential equations. Let us sketch here how to directly obtain such limit from the definition~\eqref{period}.

First observe that in the $\lambda\sim\infty$ the preimages of a point $\lambda$ split in two subsets: we denote $\zeta_a(\lambda)$ with $a=1,\dots,k$ those that tend to $\infty$ and with $\zeta_b(\lambda)$ with $b=k+1,\dots,k+m$, those that tend to $0$. 

Their asymptotic behavior for $\lambda\sim\infty$ defines two Laurent series 
\begin{align}
&\zeta_a(\lambda) \sim \lambda_a^{1/k} +O(1) \in R\lambda_a^{1/k}[[\lambda_a^{-1/k}]] \\
&\zeta_b(\lambda) e^{-t^{k+m}/m} \sim \lambda_b^{-1/m} +O(\lambda_b^{-2/m}) \in R[[\lambda_b^{-1/m}]]  \label{aympt}
\end{align}
with coefficients in $R:=\C[t^1,\dots,t^{k+m},e^{t^{k+m}/m}]$. Moreover, in the classical limit all but the leading terms in the right-hand sides of the previous expressions become zero. The indices $a$, $b$ in $\lambda_a^{1/k}$, $\lambda_b^{1/m}$ enumerate the different branches of the roots.

For $p>0$, moving the derivative $\partial_{t^\alpha}$ inside the integral and recalling $\lambda-\lambda(\zeta_a(\lambda))=0$ by definition, we can write
\begin{equation}
\label{per-1}
<I_a^{(-p)}(\lambda,t),\frac{\partial }{\partial t^\alpha} >  = \left[ d^{-1} \left( \frac{(\lambda-\lambda(\zeta))^{p-1}}{p-1!}\frac{\partial \lambda(\zeta)}{\partial t^\alpha} \frac{d\zeta}\zeta \right)  \right]_{\zeta=\zeta_a(\lambda)} .
\end{equation}
It is easy to prove that $\lambda(\zeta), \partial_{t^\alpha}\lambda(\zeta) \in R[\zeta, \zeta^{-1}]$ and that in the classical limit $\lambda(\zeta) \to \zeta^k$ and
$\frac{\partial \lambda(\zeta)}{\partial t^\alpha}$ gives $\zeta^{k-\alpha}$ for $\alpha=1,\dots,k$, and tends otherwise to zero. It follows that~\eqref{per-1} tends to
\begin{equation}
\label{per-1x}
\left[ d^{-1} \left(  \frac{(\lambda-\zeta^k)^{p-1}}{p-1!}\zeta^{k-\alpha} \frac{d\zeta}\zeta\right)  \right]_{\zeta=\lambda_a^{1/k}}
\end{equation}
for $\alpha=1,\dots,k$, and to zero for $\alpha =k+1,\dots,k+m$. 

The second set of periods corresponds to the preimages $\zeta_b(\lambda)$ that tend to $0$ for $\lambda\sim\infty$. As one can see from the asymptotic expansions~\eqref{aympt} the leading term in $\zeta_b(\lambda)$ vanishes in the classical limit. For this reason it is convenient to change variable to $\tilde\zeta = e^{-t^{k+m}/m} \zeta$ in the integration, writing $<I_b^{(-p)},\frac{\partial}{\partial t^\alpha}>$ as
\begin{equation}
\label{per-2}
\left[ d^{-1} \left( \frac{\lambda-\lambda(\tilde\zeta))^{p-1}}{p-1!} 
\left. \frac{\partial \lambda(\zeta)}{\partial t^\alpha} \right|_{\zeta=e^{t^{k+m}/m} \tilde\zeta} \frac{d\tilde\zeta}{\tilde\zeta} \right) 
\right]_{\tilde\zeta= e^{-t^{k+m}/m} \zeta_b(\lambda)} .
\end{equation}
Note that in the new variable the formal integration rules are $d^{-1}( \tilde\zeta^s d\tilde\zeta)=(s+1)^{-1}\tilde\zeta^{s+1}$ for $s\not=-1$ and $d^{-1}(\tilde\zeta^{-1}d\tilde\zeta)=\log\tilde\zeta + t^{k+m}/m$.
The classical limit of $\lambda(\tilde\zeta)$ is now given by $\tilde\zeta^{-m}$. The derivatives $\frac{\partial\lambda(\zeta)}{\partial t^\alpha}$ evaluated at $\zeta=e^{t^{k+m}/m} \tilde\zeta$ tend to $\tilde\zeta^{k-\alpha}$ for $\alpha=k,\dots,k+m$ and otherwise to $0$. Hence the classical limit of equation~\eqref{per-2} is
\begin{equation}
\label{per-2x}
\left[ d^{-1} \left(  \frac{(\lambda-\tilde\zeta^{-m})^{p-1}}{p-1!}\tilde\zeta^{k-\alpha} \frac{d\tilde\zeta}{\tilde\zeta} \right)  \right]_{\tilde\zeta=\lambda_b^{-1/m}}.
\end{equation}

Finally the integrals~\eqref{per-1x} and~\eqref{per-2x} can be easily computed in explicit form. Their generating functions
\begin{equation} \label{finf}
f_\infty^{a/b} = \sum_{n\in\Z} I^{(n)}_{a/b}(\lambda,\infty) (-z)^n
\end{equation}
are given in the next section.

\subsection{Vertex operators}

The vertex operator $\Gamma$ associated to a vector $f\in H[[z,z^{-1}]]$ is defined as
$\Gamma := e^{\widehat{f_-}} e^{\widehat{f_+}} $,
where $\widehat{f_\pm}$ are linear differential operators obtained by a quantization procedure, as described e.g. in~\cite{Giv2}. Briefly, given
\begin{equation}
f(z) = \sum_{n\in\Z} (I^{(n)})^\alpha \frac{\partial}{\partial t^\alpha} (-z)^n 
\end{equation}
the associated quantized operators are
\begin{equation}
\widehat{f_+} = \sum_{n\geq0} (I^{(n)})^\alpha (-1)^n \epsilon \frac{\partial}{\partial q^\alpha_n} , \quad
\widehat{f_-} = -\sum_{n\geq0} (I^{(-n-1)})_\alpha  \frac{q^\alpha_n}\epsilon .
\end{equation}
Here $\frac{\partial}{\partial t^\alpha}\in H$ denotes the coordinate basis of $H$ and $dt^\alpha$ the dual basis of $T^* M_{k,m}$ identified with $H$ by the metric.  
 

These operators act on the Fock space $B_H$, i.e. the space of formal functions in the variables $q_n^\alpha$ for $n\geq0$, $\alpha=1,\dots,m+k$.

By the explicit computation of the periods in the previous section we obtain two sets of vectors $f^{a/b}_\infty$. 
For $1\leq a \leq k$, we have 
\begin{align}
f_\infty^a = &\frac1k \sum_{n\geq0} \frac{\lambda^n}{n!} (\log_a \lambda - c_n) dt^k (-z)^{-n-1} 
+ \frac1k \sum_{n\geq0} n! \lambda^{-n-1} dt^k z^n \nn\\
&+ \sum_{\alpha=1}^{k-1} \sum_{n\in\Z} 
\big(\frac{\alpha}k-1\big)_n 
\lambda_a^{\alpha/k -n-1} \frac{\partial}{\partial t^\alpha}(-z)^n .
\end{align}
In this formula $\log_a \lambda$ denotes the different branches of the logarithm, parametrized by $a$, i.e. $\log_a \lambda = \log\lambda + 2\pi i (a-1)$ where $\log\lambda$ is a fixed choice of branch of the logarithm near $\lambda=\infty$. Consequently
\begin{equation}
\lambda_a^\frac1k = e^{\frac1k \log_a\lambda} = \lambda^{\frac1k} e^{2\pi i \frac{a-1}k}.
\end{equation}

We denote the harmonic numbers $c_0=0$, $c_n=1+ \cdots + \frac1n$ and for $n\in\Z$ we define the function
\begin{equation}
(q)_n := \frac{\prod_{l=-\infty}^n (q-l+1)}{\prod_{l=-\infty}^0 (q-l+1)}
=\frac{\Gamma(q+1)}{\Gamma(q-n+1)},
\end{equation}
which coincides with the Pochhammer symbol for the falling factorial for $n\geq0$.

For $k+1\leq b \leq k+m$, evaluating~\eqref{per-2x} we get
\begin{align}
f_\infty^b = &-\frac1m \sum_{n\geq0} \frac{\lambda^n}{n!} (\log_b (\lambda Q^{-m} ) -c_n) dt^k (-z)^{-n-1} - \frac1m \sum_{n\geq0} n! \lambda^{-n-1} dt^k z^n \nn\\
&- \sum_{\alpha=1}^m \sum_{n\in\Z}   
\big(\frac{\alpha}m-1\big)_n 
\lambda_b^{\alpha/m -n-1} \frac{\partial}{\partial t^{k+m-\alpha}}(-z)^n,
\end{align}
where the branch of the logarithm is parametrized by $\log_b\lambda=\log\lambda+ 2\pi i (b-k-1)$. 

We denote $\Gamma_\infty^{\pm a}$ and $\Gamma_\infty^{\pm b}$ the vertex operators corresponding to $\pm f_\infty^a$ and $\pm f_\infty^b$, respectively. 

Note that, contrary to what happens e.g. in the $A_n$ case~\cite{Giv}, the vertex operators introduced above depend not only on the roots of $\lambda$ but also on its logarithm.
As a consequence the averaging over the different branches of the roots that appears in the Hirota equation~\eqref{hqe} fails to produce a single-valued function in a neighbourhood of $\lambda\sim \infty$. Indeed the sum
\begin{equation}
\sum_{a=1}^k  \lambda_a^{\frac{1-k}k}  \Gamma_\infty^a\otimes\Gamma_\infty^{-a}
\end{equation}
gets an extra summand proportional to 
\begin{equation} \label{ext}
e^{\widehat{\phi_-}} \otimes e^{-\widehat{\phi_-}}  - 1  
\end{equation}
upon sending $\lambda$  to $\lambda e^{2\pi i}$,  where $\phi_- := 2 \pi i \sum_{n\geq0} \frac{\lambda^n}{n!} dt^k (-z)^{-n-1}$. 

The problem of the logarithmic multivaluedness has been solved by Milanov~\cite{Mil07} by introducing extra vertex operators that take values in the algebra of differential operators acting on an extra variable $x$.

Let
\begin{equation}
\Gamma_\infty^{\delta} = e^{\widehat{f_\infty^\varphi \epsilon\partial_x}} e^{\widehat{\frac{x}{\epsilon} \frac{\partial\ }{\partial t^k}}},
\qquad
\Gamma_\infty^{\delta\#} = e^{\widehat{\frac{x}{\epsilon} \frac{\partial\ }{\partial t^k}}}  e^{-\widehat{f_\infty^\varphi \epsilon\partial_x}},
\end{equation}
where
\begin{equation}
f_\infty^\varphi	= \sum_{n>0} \frac{\lambda^n}{n!} dt^k (-z)^{-n-1} \in H[[z^{-1}]] .
\end{equation}

One can check that $\Gamma_\infty^{\delta\#} \otimes \Gamma_\infty^\delta$ vanishes when composed with the extra term~\eqref{ext}, if $(q_0^k)' - (q_0^k)'' \in \Z$. A similar argument holds for the second summand in the Hirota quadratic equation~\eqref{hqe}. 
This ensures that the $1$-form in $\lambda$ appearing in the Hirota quadratic equations is single valued in a neighbourhood at $\lambda\sim\infty$.

\subsection{Hirota quadratic equations}

We say that an element $\tau$ of the Fock space $B_H$ satisfies the (descendent) Hirota quadratic equation (HQE) iff the $1$-form 
\begin{equation} \label{hqe}
(\Gamma_\infty^{\delta\#} \otimes \Gamma_\infty^\delta) 
\left(\frac1k\sum_{a=1}^k  \lambda_a^{\frac{1-k}k}  \Gamma_\infty^a\otimes\Gamma_\infty^{-a} 
-\frac{Q}m \sum_{b=k+1}^{k+m}  \lambda_b^{-\frac{1+m}m} \Gamma_\infty^b\otimes\Gamma_\infty^{-b} \right) (\tau\otimes\tau) d\lambda
\end{equation}
computed at $(q_0^k)''-(q_0^k)'=\epsilon r$ is regular in $\lambda$ for each $r\in\Z$. The tensor product $\tau\otimes\tau$ denotes the multiplication of two tau functions $\tau(q') \tau(q'')$, evaluated in variables $q'$, $q''$. 
By definition, the HQE is interpreted as follows: first we perform a change of variables $y'= \frac12(q'-q'')$ and $y''= \frac12(q'+q'')$, then we expand the result in power series in $y''$. 
By the argument in the previous section each coefficient in this expansion is a single valued function in a neighbourhood of $\lambda\sim\infty$ hence expands as a Laurent series in $\lambda^{-1}$. The requirement of regularity means that all strictly negative powers of $\lambda$ are set to zero. 

\begin{remark}
To the (calibrated) Frobenius manifold $M_{k,m}$ one associates, using Givental formula~\cite{Giv2}, an element $\cD^{M_{k,m}}$ of the Fock space $B_H$, called total descendent potential, in terms of the action of certain quantized operators on $k+m$ copies of the Kontsevich-Witten KdV tau-function. According to~\cite{MT08}, $\cD^{M_{k,m}}$ should be easily shown to coincide with the generating function of the Gromov-Witten invariants of $C_{k,m}$.

Milanov and Tseng proved in~\cite{MT08} that the total descendent potential $\cD^{M_{k,m}}$ satisfies the HQE~\eqref{hqe}. They conjectured that the HQE~\eqref{hqe} should be equivalent to the EBTH, as we prove in the following.

Note that this result follows from a similar theorem that asserts that the total ancestor potential, depending on a point $t$ of the Frobenius manifold, satisfies a $t$-dependent ancestor Hirota equation. The descendent potential is related to the ancestor potential by a lower-triangular $S$ action of the Givental group, hence the descendent Hirota equation is obtained by conjugation of the vertex operators by $S$, which is equivalent to performing a ``classical limit'' in $t$. This explains the need for the ``classical limit'' as shown above in~\ref{vertex}. 

The ancestor Hirota equations are proved by showing the regularity of the bilinear equations at the critical values $\lambda\sim u_i$ of $\lambda(\zeta)$, which, together with the property of invariance under the monodromy group of $M_{k,m}$, implies regularity at $\lambda\sim\infty$. 
\end{remark}

\section{From Hirota to Lax formulation of extended bigraded Toda hierarchy}

Let us first spell out the HQE~\eqref{hqe}. 

\subsection{Bilinear identity for the wave functions}

Let us define the following power series in $\zeta^{-1}$
\begin{equation} \label{p-tau}
\cP_1(\zeta) = \frac1\tau	
\exp\left({\frac{m}k \epsilon  \sum_{n\geq0} n! \zeta^{-nk-k}\frac{\partial\ }{\partial q_n^{k+m}} 
+ \epsilon \sum_{n\geq0} \sum_{\alpha\geq1}^{k-1} \big( n-\frac{\alpha}k\big)_n \zeta^{\alpha-nk-k} \frac{\partial\ }{\partial q_n^\alpha}}
\right)\tau
\end{equation}
such that
\begin{equation} \label{dd1}
\cP_1(\lambda_a^{1/k})  = \frac{e^{\widehat{(f^a_\infty)_+}}\tau}{\tau} 
\end{equation}
 for $1\leq a \leq k$.

Using the quantization procedure to express the vertex operators in terms of differential operators on the variables $q$, and substituting the previous definition, we get
\begin{align}
&\Gamma_\infty^{\delta\#} \Gamma_\infty^a \tau  = 
\exp ({x\frac{\partial}{\partial q_0^k}}) \cdot
  \exp \Big({\sum_{n>0} \frac{\lambda^n}{n!} q_n^k \partial_x}  \Big)\cdot \nn \\
&\cdot\exp \Big(  -\frac1{\epsilon k} \sum_{n\geq0} \big( 
\sum_{\alpha=1}^{k-1} \big( \frac{\alpha}k -1 \big)_{-n-1}
\lambda_a^{\alpha/k +n} q_n^{k-\alpha} + \frac{\lambda^n}{n!} ( \log \lambda_a - c_n) q_n^k \big) \Big) \cdot \nn\\
&\cdot \tau \cP_1(\lambda_a^{1/k}) 
\end{align}
Now we perform a careful commutation of the terms in this expression that will allow us to remove the logarithmic term. First, since in the second and third line there is no dependence on the variable $x$, we can move the second exponential to the end of the third line. 
Then we act with $\exp({x\frac{\partial}{\partial q_0^k}})$ on the quantities appearing on its right. That amounts to inserting the $x$ dependence in $\tau$ (and in $\cP_1$, since it also depends on $\tau$) and to the multiplication by a factor $\lambda_a^{-\frac{x}{\epsilon k}}$.

Denote by a prime the $x$ dependent quantities obtained by shifting $q_0^k$ by $x$, e.g.
\begin{equation}
\tau' = \tau|_{q_0^k\to q_0^k +x}, \quad \cP_1' = \cP_1|_{q_0^k\to q_0^k +x}, \quad \text{etc...}
\end{equation}

The previous expression is now equal to
\begin{align}
&\tau' \cP_1'(\lambda_a^{1/k}) 
\cdot\exp \Big(-  \frac1{\epsilon k} \sum_{n\geq0}
\sum_{\alpha=1}^{k-1}\big( \frac{\alpha}k -1 \big)_{-n-1}\lambda_a^{\alpha/k +n} q_n^{k-\alpha}\Big)\cdot\\
& \cdot \exp \Big(- \frac1{\epsilon k} \sum_{n>0} \frac{\lambda^n}{n!} ( \log \lambda_a - c_n) q_n^k  \Big) \cdot \\
&\cdot \lambda_a^{-\frac{q_0^k+x}{\epsilon k}}
\cdot\exp \Big({\sum_{n>0} \frac{\lambda^n}{n!} q_n^k \partial_x}  \Big)
\end{align}

Commuting the two terms in the last line exactly cancels the logarithm that appears in the second line. We have shown that 
\begin{align}
\Gamma_\infty^{\delta\#} \Gamma_\infty^a \tau  =
&\tau' \cP_1'(\lambda_a^{1/k}) \cdot
\exp \Big(-  \frac1{\epsilon k} \sum_{n\geq0}
\sum_{\alpha=1}^{k-1} \big( \frac{\alpha}k -1 \big)_{-n-1} \lambda_a^{\alpha/k +n} q_n^{k-\alpha}\Big)\cdot \nn\\
&\cdot \exp \Big(\frac1{\epsilon} \sum_{n>0} \frac{\lambda^n}{n!} ( \epsilon \partial_x + \frac1k c_n) q_n^k \big) \Big) \cdot
\lambda_a^{-\frac{q_0^k+x}{\epsilon k}} 
\end{align}

Defining the $\partial_x$-operator-valued wave function as
\def\cW{\mathcal{W}}
\begin{equation}
\label{cW1}
\cW_1(\zeta) = \cP_1'(\zeta) e^{- \frac1{\epsilon k} \sum_{n\geq0}
\sum_{\alpha=1}^{k-1} \big( \frac{\alpha}k -1 \big)_{-n-1} 
\zeta^{\alpha +nk} q_n^{k-\alpha} +
\frac1{\epsilon} \sum_{n>0} \frac{\zeta^{nk}}{n!} ( \epsilon \partial_x + \frac1k c_n) q_n^k  }
\end{equation}
we have
\begin{equation} \label{t1}
\Gamma_\infty^{\delta\#} \Gamma_\infty^a \tau  =
\tau' \cW_1(\lambda_a^{1/k}) \lambda_a^{-\frac{q_0^k+x}{\epsilon k}} .
\end{equation}

Note that the introduction (originally done in~\cite{Mil07}) of an differential operator-valued wave function allows to isolate the logarithmic dependence, that appears in the vertex operators, only on the last factor in~\eqref{t1}. Such factor will actually cancel in the Hirota quadratic equation. 

In a similar way we can prove that
\begin{align}
&\Gamma_\infty^\delta \Gamma_\infty^{-a} \tau = \lambda_a^{\frac{q_0^k+x}{\epsilon k}} \cW_1^* (\lambda_a^{1/k}) \tau',\\
&\Gamma_\infty^{\delta\#} \Gamma_\infty^b \tau = \tau' \cW_2(\lambda_b^{1/m}) (\lambda_b Q^{-m}) ^{\frac{q_0^k+x}{\epsilon m}},  \\
&\Gamma_\infty^{\delta} \Gamma_\infty^{-b} \tau = (\lambda_b Q^{-m}) ^{-\frac{q_0^k+x}{\epsilon m}} \cW^*_2(\lambda_b^{1/m}) \tau' .
\end{align}
where
\begin{align}
&\cW_1^{*}(\zeta) = e^{\frac1{\epsilon k} \sum_{n\geq0}
\sum_{\alpha=1}^{k-1} \big( \frac{\alpha}k -1 \big)_{-n-1} 
\zeta^{\alpha +nk} q_n^{k-\alpha} -
\frac1{\epsilon} \sum_{n>0} \frac{\zeta^{nk}}{n!} ( \epsilon \partial_x + \frac1k c_n) q_n^k  } {\cP_1^*}' (\zeta) ,\nn\\
&\cW_2(\zeta) = \cP_2'(\zeta) e^{ \frac1{\epsilon m} \sum_{n\geq0}
\sum_{\alpha=1}^{m} \big( \frac{\alpha}m -1 \big)_{-n-1} 
\zeta^{\alpha +nm} q_n^{k+\alpha} +
\frac1{\epsilon} \sum_{n>0} \frac{\zeta^{nm}}{n!} ( \epsilon \partial_x - \frac1m c_n) q_n^k  } ,\nn\\
&\cW_2^{*} (\zeta) = e^{ -\frac1{\epsilon m} \sum_{n\geq0}
\sum_{\alpha=1}^{m} \big( \frac{\alpha}m -1 \big)_{-n-1} 
\zeta^{\alpha +nm} q_n^{k+\alpha} -
\frac1{\epsilon} \sum_{n>0} \frac{\zeta^{nm}}{n!} ( \epsilon \partial_x - \frac1m c_n) q_n^k  } {\cP^*_2}'(\zeta) .
\label{cWrest}
\end{align}

The remaining symbols of the dressing operators are define as
\begin{subequations}
\begin{align}
&\cP_1^*(\zeta) = \frac1\tau	
\exp\left({-\frac{m}k \epsilon  \sum_{n\geq0} n! \zeta^{-nk-k}\frac{\partial\ }{\partial q_n^{k+m}} 
- \epsilon \sum_{n\geq0} \sum_{\alpha\geq1}^{k-1} \big( n-\frac{\alpha}k\big)_n \zeta^{\alpha-nk-k} \frac{\partial\ }{\partial q_n^\alpha}}
\right)\tau , \\
&\cP_2(\zeta) =  \frac1\tau	
\exp\left({- \epsilon  \sum_{n\geq0} n! \zeta^{-nm-m}\frac{\partial\ }{\partial q_n^{k+m}} 
- \epsilon \sum_{n\geq0} \sum_{\alpha\geq1}^{m} \big(n-\frac{\alpha}m\big)_n \zeta^{\alpha-nm-m}  \frac{\partial\ }{\partial q_n^{k+m-\alpha}}}
\right)\tau , \label{p-tau2}
\\
&\cP_2^*(\zeta) =  \frac1\tau	
\exp\left({\epsilon  \sum_{n\geq0} n! \zeta^{-nm-m}\frac{\partial\ }{\partial q_n^{k+m}} 
+ \epsilon \sum_{n\geq0} \sum_{\alpha\geq1}^{m} \big(n-\frac{\alpha}m\big)_n \zeta^{\alpha-nm-m}  \frac{\partial\ }{\partial q_n^{k+m-\alpha}}}
\right)\tau 
\end{align}
\end{subequations}
in such a way that the following expressions, analogous to~\eqref{dd1}, hold
\begin{align}
&\cP_1^*(\lambda_a^{1/k}) = \frac{e^{-\widehat{(f^a_\infty)_+}}\tau}{\tau} ,
\\
&\cP_2(\lambda_b^{1/m}) = \frac{e^{\widehat{(f^b_\infty)_+}}\tau}{\tau} ,
 \\ 
&\cP_2^*(\lambda_b^{1/m}) =\frac{e^{-\widehat{(f^b_\infty)_+}}\tau}{\tau}.
\end{align}

Substituting this in the HQE~\eqref{hqe} and multiplying on the left by $\tau'(q')^{-1}$, and on the right by $\tau'(q'')^{-1}$, we obtain the following equivalent HQE
\begin{align}
&\frac1k \sum_{a=1}^k \lambda_a^{(1-k)/k} \cW_1(\lambda_a^{1/k}) \lambda_a^{-\frac{(q_0^k)'  +x}{\epsilon k}} \cdot \lambda_a^{\frac{(q_0^k)'' +x}{\epsilon k}} \cW_1^* (\lambda_a^{1/k}) d\lambda - \nn\\
-\frac{Q}m &\sum_{b=k+1}^{k+m} \lambda_b^{-\frac{1+m}m} \cW_2(\lambda_b^{\frac1m}) (\lambda_b Q^{-m})^{\frac{(q_0^k)' +x}{\epsilon m}} (\lambda_b Q^{-m})^{-\frac{(q_0^k)'' +x}{\epsilon m}}
\cW^*_2(\lambda_b^{\frac1m})
d\lambda \nn .
\end{align}
In this formula $\cW_1$, $\cW_2$ are evaluated in the variables $q'$, while $\cW_1^*$, $\cW_2^*$ are evaluated in the variables $q''$.

Recalling moreover that $(q_0^k)'' -(q_0^k)' = \epsilon r$ for $r\in \Z$, the previous expression becomes
\begin{align} \label{eq36}
\frac1k \sum_{a=1}^k \cW_1(\lambda_a^{1/k})  \cW_1^* (\lambda_a^{1/k}) \lambda_a^{(r-k+1)/k}  d\lambda - \notag\\
-\frac{Q^{r+1}}m \sum_{b=k+1}^{k+m} 
\cW_2(\lambda_b^{\frac1m}) \cW^*_2(\lambda_b^{\frac1m}) 
\lambda_b^{(-m-r-1)/m} d\lambda .
\end{align}

The two lines in this expression are formal series in $\lambda_a^{\frac1k}$ and $\lambda_b^{\frac1m}$, respectively. The averages over the $k$-th and $m$-th roots of unity, respectively, ensure that the non integer roots cancel, i.e. it is a formal series in integer powers of $\lambda$. 

By a change of variable, we can easily show that the regularity of $1$-form~\eqref{eq36} is equivalent to the following residue formula
\begin{align} \label{hbe-wave}
\res_\zeta \cW_1(\zeta) \cW_1^*(\zeta) \zeta^{ks+r} {d\zeta}=
Q^{r+1} 
\res_\zeta \cW_2(\zeta) \cW_2^*(\zeta) \zeta^{ms-r-2} {d\zeta}
\end{align}
for each $s\geq0$. As above, $\cW_1$, $\cW_2$ are evaluated in the variables $q'$, while $\cW_1^*$, $\cW_2^*$ are evaluated in the variables $q''$, and $(q_0^k)'' -(q_0^k)' = \epsilon r$ for $r\in \Z$.

\subsection{Difference operators}
\label{s:difference}

This bilinear expression can be reformulated in terms of difference operators obtained by ``quantizing'' the symbols $\cW_i$. 
Let $A = \sum_s a_s \Lambda^s$ be a difference operator, where the coefficients $a_s$ are functions of $x$ and the shift operator $\Lambda^s$ acts as $\Lambda^s a(x) = a(x+\epsilon s) \Lambda^s$. The left  and right symbols of $A$ are formal functions of $\zeta$ defined as 
\begin{equation}
\sigma_l(A)=\sum_s a_s\zeta^s, \quad
\sigma_r(A)=\sum_s \tilde{a}_s \zeta^s
\end{equation}
where the coefficients $\tilde{a}_s(x)=a_s(x-\epsilon s)$ are such that $A=\sum_s \Lambda^s \tilde{a}_s$. 
Note that we will deal with symbols whose coefficients are differential operators in $x$.

To reconstruct from~\eqref{hbe-wave} a bilinear expression in terms of difference operators we need to decide which expressions are associated to left and right symbols. In particular we have to take care of the fact that the variable $x$ in $\cW_i$ is shifted by $\epsilon r$ and the same constant $r$ appears in the exponent of $\zeta$ in the residues with opposite sign.

Let us define difference wave-operators $W_i$ and $W_i^*$ such that 
\begin{subequations}
\label{wavesym}
\begin{align}
&\sigma_l (W_1) = \cW_1(\zeta^{-1}), & &\sigma_r(W_1^*) = \cW_1^*(\zeta^{-1}), \\
&\sigma_l(W_2) =\cW_2(Q \zeta), & &\sigma_r(W_2^*) = \cW_2^*(Q\zeta).
\end{align} 
\end{subequations}

Let us consider the left-hand side of~\eqref{hbe-wave}. 
Changing the integration variable to $\zeta^{-1}$ and 
substituting the definitions of the wave-operators we get
\begin{equation}
\res_\zeta \sigma_l(W_1\Lambda^{-ks-1}) \sigma_r(W_1^* \Lambda^{-r}) \frac{d\zeta}\zeta .
\end{equation}
Note that we have used the obvious identities $\sigma_l(A \Lambda^r) = \sigma_l(A) \zeta^r$ and $\sigma_r(A \Lambda^r) = \sigma_r(A(x-\epsilon r)) \zeta^r$ valid for any difference operator $A$.

Defining $\res_\Lambda A = a_0$ for $A=\sum_s a_s \Lambda^s$, one can easily check that
\begin{equation}
\res_\zeta \sigma_l(A) \sigma_r(B) \frac{d\zeta}\zeta = \res_\Lambda AB
\end{equation}
for any two difference operators $A$ and $B$.

The left-hand side in~\eqref{hbe-wave} is then equal to
\begin{equation}
\res_\Lambda W_1 \Lambda^{-ks-1} W_1^* \Lambda^{-r}
\end{equation}
and a similar computation gives that the right-hand side (after rescaling $\zeta \to Q \zeta$) is equal to 
\begin{equation}
\res_\Lambda Q^{ms} W_2 \Lambda^{ms-1} W_2^* \Lambda^{-r} .
\end{equation}
Since these expressions have to be equal for any value of $r\in \Z$, we  obtain that the bilinear equation~\eqref{hbe-wave} is equivalent to the following identity of difference operators for $s\geq0$
\begin{equation} \label{hbe-oper}
W_1(q',x) \Lambda^{-ks -1}W_1^*(q'',x) = Q W_2(q',x) (Q \Lambda)^{ms-1} W_2^*(q'',x).
\end{equation}
Note that in this expression we set $(q_0^k)' = (q_0^k)''$ and that the coefficients expand as differential operators in $x$.

\subsection{Sato-Wilson and Lax equations}

Now we examine some consequences of the last equation.
First observe that by definition the operators $W_i$ and $W_i^*$ are of the following form
\begin{subequations}
\label{WWWW}
\begin{align}
W_1 &= P_1 e^{- \frac1{\epsilon k} \sum_{n\geq0}
\sum_{\alpha=1}^{k-1} \big( \frac{\alpha}k -1 \big)_{-n-1} 
\Lambda^{-\alpha -nk} q_n^{k-\alpha}
 +
\frac1{\epsilon} \sum_{n>0} \frac{\Lambda^{-nk}}{n!} ( \epsilon \partial_x + \frac1k c_n) q_n^k 
 }
\\
W_2 &= P_2 
e^{ \frac1{\epsilon m} \sum_{n\geq0}
\sum_{\alpha=1}^{m} \big( \frac{\alpha}m -1 \big)_{-n-1} 
(Q \Lambda)^{\alpha +nm} q_n^{k+\alpha}
 +
\frac1{\epsilon} \sum_{n>0} \frac{(Q\Lambda)^{nm}}{n!} ( \epsilon \partial_x - \frac1m c_n) q_n^k 
 }\\
W_1^* &=  e^{\frac1{\epsilon k} \sum_{n\geq0}
\sum_{\alpha=1}^{k-1} \big( \frac{\alpha}k -1 \big)_{-n-1} 
\Lambda^{-\alpha -nk} q_n^{k-\alpha}-\frac1{\epsilon} \sum_{n>0} \frac{\Lambda^{-nk}}{n!} ( \epsilon \partial_x +\frac1k c_n) q_n^k }
P_1^*
\\
W_2^{*}& =e^{ -\frac1{\epsilon m} \sum_{n\geq0}
\sum_{\alpha=1}^{m} \big( \frac{\alpha}m -1 \big)_{-n-1} 
(Q\Lambda)^{\alpha +nm} q_n^{k+\alpha}-\frac1{\epsilon} \sum_{n>0} \frac{(Q\Lambda)^{nm}}{n!} ( \epsilon \partial_x - \frac1m c_n) q_n^k }
P^*_2
\end{align}
\end{subequations}
where $P_i$ and $P_i^*$ are difference operators given by
\begin{align}
&\sigma_l(P_1) =\cP_1'(\zeta^{-1}), 
& &\sigma_r(P_1^*) = {\cP_1^*}'(\zeta^{-1}) \\
&\sigma_l(P_2) = \cP_2'(Q \zeta), 
& &\sigma_r(P_2^*) ={ \cP_2^*}'(Q\zeta).
\end{align}
In particular the operators $P_i$ are of the form
\begin{subequations}
\label{PP}
\begin{align}
&P_1 = 1+ p_{1} \Lambda + \dots \\
&P_2 = p_0 + p_{-1} \Lambda^{-1} + \dots ,
\end{align}
\end{subequations}
where the coefficients $p_i$ are expressed in terms of the tau function by the definition~\eqref{p-tau} and~\eqref{p-tau2}. Similar formulas hold for $P_i^*$.
Note that the exponentials in $W_1^*$, $W_2^*$ are the inverses of the exponentials appearing in $W_1$, $W_2$, respectively.

Evaluating the bilinear identity for the wave operators~\eqref{hbe-oper} at different values of $s$ and setting $q''=q'$ we can obtain the Sato-Wilson equation and the Lax formulation of the hierarchy. 

Let $s=0$ and $q'=q''$. From~\eqref{hbe-oper} we get
\begin{equation}
P_1 \Lambda^{-1}	P_1^* = P_2  \Lambda^{-1} P_2^* .
\end{equation}
It easily follows that $P_1^*(x-\epsilon)$ and $P_2^*(x-\epsilon)$ are the inverse operators to $P_1$ and $P_2$, respectively.

Using this fact and letting $s=1$ and $q'=q''$ we obtain
\begin{equation}
\label{L=L}
P_1 \Lambda^{-k} P_1^{-1} = P_2   (Q\Lambda)^{m}P_2^{-1} =: L .
\end{equation}
Clearly the operator $L$ is of the form 
\begin{equation} \label{Laxop}
L = \Lambda^{-k} + \dots + u_{m} \Lambda^{m} .
\end{equation}

Finally, differentiate (\ref{hbe-oper}) to $q_n^{k-\alpha}$ for $1\le \alpha<k$, respectively, to $q_n^{k+\alpha}$ for $1\le \alpha\le m$ and to $q_n^k$, this gives, using the above considerations
\begin{align}
&\frac{\partial P_1 }{\partial  q_n^{k-\alpha}}	P_1^{-1} 
-\frac{1}{\epsilon k}(\frac{\alpha}{k}-1)_{-n-1}P_1\Lambda^{-\alpha-nk}P_1^{-1} 
=\frac{\partial P_2  }{\partial  q_n^{k-\alpha}} P_2^{-1},\label{L58}
\\
&\frac{\partial P_1 }{\partial  q_n^{k+\alpha}}	P_1^{-1} 
=\frac{\partial P_2  }{\partial  q_n^{k+\alpha}}
 P_2^{-1}
 +\frac{1}{\epsilon m}(\frac{\alpha}{m}-1)_{-n-1}
P_2  
(Q\Lambda)^{\alpha+nm}P_2^{-1},\label{L59}\\
&\frac{\partial P_1 }{\partial  q_n^{k}}P_1^{-1} 
+\frac{1}{\epsilon n!}P_1(\epsilon\partial_x+\frac1{k}c_n)\Lambda^{-nk}
P_1^{-1} =\nn\\
&\qquad\qquad=\frac{\partial P_2  }{\partial  q_n^{k}} P_2^{-1}
+\frac{1}{\epsilon n!}P_2 (\epsilon\partial_x-\frac1{m}c_n)(Q\Lambda)^{nm}P_2^{-1}.\label{L60}
\end{align}
Projecting on the positive ($>0$), respectively non-positive ($\le 0$), degrees of $\Lambda$, one obtains the following Sato-Wilson equations
\begin{equation}
\label{Sato-Wilson}
\frac{\partial P_1 }{\partial  q_n^{\beta}}	P_1^{-1} 
=\left(B_n^{\beta}\right)_{>0},\quad
\frac{\partial P_2  }{\partial  q_n^{\beta}} P_2^{-1}
=
-\left(B_n^{\beta}\right)_{\le 0},
\end{equation}
where $B_n^\beta$ is defined by
\begin{align}
B_n^\beta&=\frac{1}{\epsilon k}(-\frac{\beta}{k})_{-n-1}L^{n-\frac{\beta}k+1} \qquad\qquad\quad1\leq\beta\leq k-1, \label{Bs1}\\
B_n^\beta&=\frac{1}{\epsilon m}(\frac{\beta-k}{m}-1)_{-n-1}L^{\frac{\beta-k}m+n} \qquad k+1\leq\beta\leq k+m, \label{Bs2}\\
B_n^{k}&=\frac{2}{\epsilon n!}\left(L^n\left(
\log L -\frac{c_n}2 (\frac1k+\frac1m)-\frac12 \log Q\right)\right).
\label{Bs3}
\end{align}

We have defined the following operators as  in~\cite{Car06}: the roots of the Lax operator $L$
\[
L^{\frac1k}=P_1\Lambda^{-1} P_1^{-1} ,\quad
L^{\frac1m}=P_2 (Q\Lambda)P_2^{-1}
\]
and its logarithm
\[
\log L=\frac1{2m}\log_{+} L+\frac1{2k}\log_{-} L,
\]
where
\begin{align}
\log_{-} L&=-kP_1\epsilon\partial_x P_1^{-1} ,\\
\log_{+} L&=mP_2 \epsilon\partial_xP_2^{-1}+m\log Q.
\end{align}

We see that starting from the Hirota quadratic equation, we are naturally led to the definition of logarithms of $L$. They are difference operators which contain the operator of derivation in $x$. 
See~\cite{CDZ, Car06} for further discussion of their properties. 

The Lax equations are easily derived from the Sato-Wilson equations:
\begin{equation} \label{laxxx}
\frac{\partial L}{\partial  q_n^{\beta}}
=-\left[\left(B_n^{\beta}\right)_{\le 0},L\right]=\left[\left(B_n^{\beta}\right)_{>0},L\right] .
\end{equation}
These Lax equations, together with the form of $L$ given in~\eqref{Laxop} and the definitions of its roots and logarithms, defines a hierarchy called (extended) bigraded Toda hierarchy, which has been introduced in~\cite{Car06}. We have shown that 
\begin{proposition}
The Lax operator $L$ associated to a tau function which satisfies the Hirota quadratic equation~\eqref{hqe} is a solution of the $(k,m)$-extended bigraded Toda hierarchy Lax equations~\eqref{laxxx}. 
\end{proposition}

As a corollary, it follows from Milanov-Tseng theorem~\cite{MT08} that the total descendent potential of $\cC_{k,m}$ is a tau function of this hierarchy. 

\begin{remark}
The bigraded Toda hierarchy has been defined in~\cite{Car06} in a slightly different way. The two formulations are identified by changing $\epsilon\to-\epsilon$, and correspondingly $\Lambda\to\Lambda^{-1}$, and by rescaling the times $q_n^\beta \to s_\beta q_n^\beta$, with constants $s_\beta$ not dependent on $n$.
\end{remark}

\section{From Lax to Hirota}
We have shown how to derive the Lax formulation of the bigraded Toda hierarchy from the Hirota quadratic equations. 
The Hirota  equations are actually equivalent to the Lax formulation, or rather to the Sato-Wilson equations, of the EBTH. In this section we will briefly recall the construction of the tau function starting from a solution of the EBTH, and sketch the proof that the Hirota quadratic equations are satisfied. 

We say that two operators $P_1$ and $P_2$ of the form~\eqref{PP} are dressing operators for the EBTH hierarchy if they satisfy the Sato-Wilson equations~\eqref{Sato-Wilson} and the constraint~\eqref{L=L}. 
The corresponding wave functions $\cW_i$, $\cW^*_i$ are defined as the symbols~\eqref{wavesym} of the wave operators~\eqref{WWWW}.

\begin{proposition}
The operators $P_1$, $P_2$ are dressing operators for the EBTH hierarchy if and only if the corresponding wave functions satisfy the bilinear equation~\eqref{hbe-wave}.
\end{proposition}
\begin{proof}
The proof is quite standard (see for example~\cite{UT84, Mil07, LHWC10}). We have already shown that the bilinear equations for the wave functions imply the Sato-Wilson equations. Let us sketch the proof of the converse.

Let $P_1$ and $P_2$ be dressing operators for EBTH. Define the wave operators $W_1$, $W_2$, $W^*_1$ and $W_2^*$ by formula~\eqref{WWWW} and recall that $P_i^* \Lambda = \Lambda P_i^{-1}$, $i=1,2$.

The first important observation is that the wave operators satisfy the same equations in $q_n^\beta$, but for the case $\beta=k$, $n=0$. It follows that
\begin{equation}
\frac{\partial W_1}{\partial q_n^\beta} W_1^{-1} = \frac{\partial W_2}{\partial q_n^\beta} W_2^{-1},
\end{equation}
for $(\beta,n)\not=(k,0)$.
It is also quite obvious that 
\begin{equation}
W_1 \Lambda^{-1} W_1^* = Q W_2 (Q\Lambda)^{-1} W_2^* .
\end{equation}
It is then a simple matter of induction to prove that
\begin{equation}
\left(\frac{\partial}{\partial q^{\alpha_1}_{n_1}} \dots \frac{\partial}{\partial q^{\alpha_\ell}_{n_\ell}} W_1\right) \Lambda^{-1} W_1^* 
= Q \left(\frac{\partial}{\partial q^{\alpha_1}_{n_1}} \dots \frac{\partial}{\partial q^{\alpha_\ell}_{n_\ell}} W_2 \right)(Q \Lambda)^{-1} W_2^*
\end{equation}
for any multi-index $(\alpha_1,n_1; \dots; \alpha_\ell,n_\ell)$, for all pairs of indexes, excluding $(\alpha_i,n_i)=(k,0)$. From this the bilinear equation for the wave operators~\eqref{hbe-oper} with $s=0$ simply follows, using a Taylor series expansion. Multiplying on the left by $L^s$, we obtain~\eqref{hbe-oper} for any $s\geq0$.

Inverting the argument used in section~\ref{s:difference}, it is clear that~\eqref{hbe-oper} is equivalent to~\eqref{hbe-wave}.
\end{proof}

From this proof it is clear that the Sato-Wilson equations are also equivalent to the bilinear equation for the wave operators~\eqref{hbe-oper}.

We now want to show the existence of a tau function, namely that, given dressing operators $P_1$ and $P_2$ of the EBTH, one can find a function $\tau$, depending on the times $q_n^\alpha$ and the dispersive parameter $\epsilon$ such that~(\ref{p-tau}, \ref{p-tau2}) hold. These equations can be written as
\begin{equation} \label{PPP} 
\cP_i(\zeta)=\frac{\tau(q-[\zeta^{-1}]_i)}{\tau(q)}, \quad i=1,2
\end{equation}
where the shift functions $[\zeta^{-1}]_1$ and $[\zeta^{-1}]_2$ can be read off from~(\ref{p-tau}, \ref{p-tau2}).

The main observation is that these shifts do not involve the ``logarithmic'' variables $q_n^k$ for $n>0$. (Note that in~\eqref{p-tau2} the coefficient $(n-1)_n$ in the term containing $\frac{\partial}{\partial q_n^k}$ is equal to $\delta_{n,0}$.) 
Therefore, the times $q_n^k$ for $n>0$ enter only as parameters in equations~\eqref{PPP}.

To prove~\eqref{PPP} we need to show that some compatibility conditions hold. Such compatibility equations will not involve the times $q_n^k$ explicitly. Since the only times that do not come from a reduction of the 2D Toda times are precisely the logarithmic ones, the proof of the compatibility equations can be performed exactly as in the 2D Toda case~\cite{UT84}.

Another consequence of the fact that we do not shift the ``logarithmic'' times $q_n^k$ ($n> 0$) in~\eqref{PPP}, is that the tau function has a much bigger arbitrarity than in the usual (non-extended) case.

\begin{proposition}
\label{propPtau}
Let $P_1$ and $P_2$ be dressing operators for the extended bigraded Toda hierarchy, then there exists a function $\tau$ such that~\eqref{PPP} holds. The function $\tau$ is uniquely determined up to right multiplication by  a non-vanishing function depending only on $q_n^k$, $n>0$.
\end{proposition}
\begin{proof}

The compatibility conditions for the formulas~\eqref{PPP} are
\begin{equation} \label{compat}
\cP_i(q-[\xi^{-1}]_j;\zeta) \cP_j(q;\xi) = 
\cP_j(q-[\zeta^{-1}]_i;\xi) \cP_i(q;\zeta)
\end{equation}
for $i,j=1,2$.

One can obtain these equations from  (\ref{hbe-oper}) in a similar way as is done in \cite{UT84}. We refer to that paper for the details of a proof, here we will only give a sketch. We first rewrite (\ref{hbe-oper}) by introducing some arbitrary parameters $t_j$:
\begin{equation}
\label{hbe-oper2}
W_1(q')e^{\sum_{j=1}^\infty t_j\Lambda^{-kj}}\Lambda^{-1} W_1^*(q'')=
W_2(q')e^{\sum_{j=1}^\infty t_j(Q\Lambda)^{mj}}\Lambda^{-1} W_2^*(q'')
\end{equation}
Substitute $q''=q'-[\zeta^{-1}]_1-[\xi^{-1}]_1$ and $t_j=\frac{\zeta^{-jk}+\xi^{-jk}}{jk}$ in (\ref{hbe-oper2}), this gives equation (\ref{lemma1a}) of

\begin{lemma}
\label{lemma1}
Let $(1-z\Lambda)^{-1}=\sum_{n=0}^\infty \Lambda^nz^n$, the following identities hold:
\begin{align}
\label{lemma1a}
&P_1(q)(1-(\zeta\Lambda)^{-1})^{-1}(1-(\xi\Lambda)^{-1})^{-1} \Lambda^{-1} P_1^*(q-[\zeta^{-1}]_1-[\xi^{-1}]_1)=
\nn\\
&\qquad=P_2(q) \Lambda^{-1}P_2^*(q-[\zeta^{-1}]_1-[\xi^{-1}]_1),
\\
\label{lemma1b}
&P_1(q)(1-(\zeta\Lambda)^{-1})^{-1}\Lambda^{-1} P_1^*(q-[\zeta^{-1}]_1-[\xi^{-1}]_2)=
\nn\\
&\qquad =P_2(q)
\left(1-\frac{\Lambda Q}{\xi}\right)^{-1} \Lambda^{-1}
P_2^*(q-[\zeta^{-1}]_1-[\xi^{-1}]_2),
\\
\label{lemma1c}
&P_1(q) \Lambda^{-1} P_1^*(q-[\zeta^{-1}]_2-[\xi^{-1}]_2)=
\nn\\
&\qquad =P_2(q)
\left(1-\frac{\Lambda Q}{\zeta}\right)^{-1}\left(1-\frac{\Lambda Q}{\xi}\right)^{-1} \Lambda^{-1}
P_2^*(q-[\zeta^{-1}]_2-[\xi^{-1}]_2).
\end{align}
\end{lemma}

The other two formulas can be obtained in a similar way. Next, using
\[
(1-(\zeta\Lambda)^{-1})^{-1}(1-(\xi\Lambda)^{-1})^{-1} \Lambda^{-1} 
=\frac{\zeta\xi}{\xi-\zeta}
\left(
(1-(\zeta\Lambda)^{-1})^{-1}-(1-(\xi\Lambda)^{-1})^{-1}
\right)
\]
and  taking the residue in (\ref{lemma1a}) one obtains the first equation in
\begin{lemma}The following identities hold:
\label{lemma2}
\begin{align}
\label{lemma2a}
&\cP_1(x,q,\zeta){\cP_1^*}(x,q-[\zeta^{-1}]_1-[\xi^{-1}]_1,\zeta)=
\nn\\&=
\cP_1(x,q,\xi){\cP_1^*}(x,q-[\zeta^{-1}]_1-[\xi^{-1}]_1,\xi),
\\
\label{lemma2b}
&\cP_1(x,q,\zeta){\cP_1^*}(x-\epsilon,q-[\zeta^{-1}]_1-[\xi^{-1}]_2,\zeta)=
\nn\\&=
\cP_2(x,q,\xi){\cP_2^*}(x-\epsilon,q-[\zeta^{-1}]_1-[\xi^{-1}]_2,\xi),
\\
\label{lemma2c}
&\cP_2(x,q,\zeta){\cP_2^*}(x-2\epsilon,q-[\zeta^{-1}]_2-[\xi^{-1}]_2,\zeta)=
\nn\\&=
\cP_2(x,q,\xi){\cP_2^*}(x-2\epsilon,q-[\zeta^{-1}]_2-[\xi^{-1}]_2,\xi).
\end{align}
\end{lemma}

Setting $\zeta=\infty$ in (\ref{lemma2a}) and  (\ref{lemma2b}) gives
\begin{equation}
\label{Pinverse}
\cP_1(x,q,\xi){\cP_1^*}(x,q-[\xi^{-1}]_1,\xi)=1=\cP_2(x,q,\xi){\cP_2^*}(x-\epsilon,q-[\xi^{-1}]_2,\xi).
\end{equation}
Using this formula to eliminate ${\cP_1^*}$ and ${\cP_2^*}$ in the last lemma, we obtain the compatibility conditions~\eqref{compat}.
\end{proof}

Using (\ref{Pinverse}), one deduces from~\eqref{PPP} the following
\begin{corollary}
\begin{equation}
\label{Pstar in tau}
{\cP_1^*}(\zeta)=\frac{\tau(q+[\zeta^{-1}]_1)}{\tau(q)},\qquad
{\cP_2^*}'(x,q,\zeta)=\frac{\tau(x+\epsilon,q+[\zeta^{-1}]_2)}{\tau(x,q)}.
\end{equation}
\end{corollary}

\begin{remark}
The construction in this section, which follows the method of~\cite{Mil07}, was done also by Li~et~al. in~\cite{LHWC10}, but, as they point out, it doesn't allow them to reproduce the form of the Milanov-Tseng Hirota equations.  This is mainly due to the fact that their choice of wave functions, and consequently, of vertex operator was not consistent with that of~\cite{MT08}.
\end{remark}

\section{On an alternative formulation of EBTH Hirota equations}

In this section we comment on a different formulation of the Hirota quadratic equations for the EBTH, which does not involve vertex operators with coefficients in the algebra of differential operators in $x$, and was first proposed by Takasaki in~\cite{Tak10}.

As recently observed by B.~Bakalov~\cite{Bak}, the operator 
\begin{equation}
N = e^{ - \sum_{n>0} \frac{\lambda^n}{n!} q_n^k \partial_{q_0^k} } \otimes e^{ - \sum_{n>0} \frac{\lambda^n}{n!} q_n^k \partial_{q_0^k} } 
\end{equation}
enjoys the same property as $\Gamma^{\delta\#}_\infty \otimes \Gamma^\delta_\infty$ of killing the monodromy term~\eqref{ext}. This points to the following alternative form of the Hirota quadratic equation:
\begin{equation} \label{hqe-alt}
N 
\left(\frac1k\sum_{a=1}^k  \lambda_a^{\frac{1-k}k}  \Gamma_\infty^a\otimes\Gamma_\infty^{-a} 
-\frac{Q}m \sum_{b=k+1}^{k+m}  \lambda_b^{-\frac{1+m}m} \Gamma_\infty^b\otimes\Gamma_\infty^{-b} \right) (\tau\otimes\tau) d\lambda .
\end{equation}
As before we say that this HQE is satisfied if the expansion at $\lambda \sim \infty$ does not contain negative powers in $\lambda$, i.e. if the $1$-form~\eqref{hqe-alt} is regular, for $(q''-q')_0^k = \epsilon r\in\epsilon\Z$.
Note that~\eqref{hqe-alt} does not depend so far on $x$, but such dependence can be easily added by shifting both $(q_0^k)'$ and $(q_0^k)''$ by $x$.

Let us first obtain an explicit form for this HQE. Spelling out the regularity condition above in terms of residues, similarly to what was done in section~2, we can easily see that the HQE~\eqref{hqe-alt} is satisfied iff the following residue formula holds:
{\Small
\begin{align}
&\res \zeta^{k(l-1)+r} \exp \frac1{\epsilon k} \sum_{n\geq0}  \left(\sum_{\alpha=1}^{k-1} \left( \frac\alpha{k} -1 \right)_{-n-1} \zeta^{\alpha+kn} (q''-q')_n^{k-\alpha}   +  \frac1{n!} \zeta^{nk} c_n (q'-q'')_n^k \right)\cdot \notag\\
&\quad\cdot \tau\left( q' - [\zeta^{-1}]_1, (q_0^k)'-\sum_{n>0} \frac{\zeta^{kn}}{n!} (q_n^k)' \right)  \cdot \tau\left( q'' + [\zeta^{-1}]_1, (q_0^k)''-\sum_{n>0} \frac{\zeta^{kn}}{n!} (q_n^k)'' \right) d\zeta = \\
&= \res \zeta^{m(l-1) -2 -r} Q^{r+1} 
\exp \frac1{\epsilon m} \sum_{n\geq0} \left( \sum_{\alpha=1}^m \left( \frac\alpha{m} -1 \right)_{-n-1} \zeta^{\alpha+nm} (q'-q'')_n^{\alpha+k} +  \frac{\zeta^{nm}}{n!} c_n (q''-q')_n^k  \right) \cdot \\
&\quad\cdot \tau\left( q' - [\zeta^{-1}]_2, (q_0^k)'-\sum_{n>0} \frac{\zeta^{mn}}{n!} (q_n^k)' \right) \cdot \tau\left( q'' + [\zeta^{-1}]_2, (q_0^k)''-\sum_{n>0} \frac{\zeta^{mn}}{n!} (q_n^k)'' \right) d\zeta
\end{align}
}
where $(q''-q')_0^k = \epsilon r$ with $r\in\Z$ and $l\geq1$. 

One can easily be convinced that in this form this equation is in principle equivalent to the Hirota equation proposed, for the case $k=m=1$, in~\cite{Tak10}.

It remains to show that~\eqref{hqe-alt} is indeed equivalent to~\eqref{hqe}.
Let us denote the $1$-form~\eqref{hqe-alt} by $N (\omega)$. The HQE~\eqref{hqe-alt} is equivalent to the regularity of $N(\omega)$ while the HQE~\eqref{hqe} is equivalent to the regularity of $(\Gamma^{\delta\#}_\infty \otimes \Gamma^\delta_\infty)(\omega)$, in both cases evaluated at $(q''-q')_0^k = \epsilon r\in\epsilon\Z$. One can easily check that the following identity holds
\begin{equation}
\Gamma^{\delta\#}_\infty \otimes \Gamma^\delta_\infty = ( e^{\sum_{n>0} \frac{\lambda^n}{n!} q_n^k \partial_x} e^{x \partial_{q_0^k} } \otimes  e^{x \partial_{q_0^k} } e^{-\sum_{n>0} \frac{\lambda^n}{n!} q_n^k \partial_x} )  N.
\end{equation}
It follows that the HQE~\eqref{hqe} is equivalent to the regularity of 
\begin{equation}
e^{\sum_{n>0} \frac{\lambda^n}{n!}(q_n^k)' \partial_x} N(\omega)|_x e^{-\sum_{n>0} \frac{\lambda^n}{n!}(q_n^k)'' \partial_x} .
\end{equation}
The subscript $x$ denotes that we have inserted the dependence on $x$ by shifting shifting both $(q_0^k)'$ and $(q_0^k)''$ by $x$. Since both left and right multiplication by any operator depending only on positive powers of $\lambda$ preserves the regularity of a $1$-form, we can conclude that~\eqref{hqe} and~\eqref{hqe-alt} are indeed equivalent.

\end{document}